\documentclass[paper]{ieice}
\usepackage{graphicx}
\usepackage{latexsym}
\usepackage[fleqn]{amsmath}
\usepackage{bm}
\usepackage{theorem}
\usepackage{listings}
\usepackage{color}
\usepackage{url}
\usepackage{subfigure}
\usepackage{multirow}
\usepackage{bussproofs}
\usepackage{algorithm}
\usepackage{algorithmic}
\usepackage{intmacros}
\usepackage{wrapfig}

\usepackage{article}

\newcommand{\Todo}[1]{}

\newcommand{\Until}{\mathsf{U}}
\newcommand{\Always}{\mathsf{G}}
\newcommand{\Eventually}{\mathsf{F}}

\newtheorem{theorem}{Theorem}
\newtheorem{lemma}{Lemma}

\newtheorem{definition}{Definition}
\newtheorem{example}{Example}

\newenvironment{proof}[1][Proof.]{\begin{trivlist}
\item[\hskip \labelsep {\itshape #1}]}{\end{trivlist}}

\newcommand{\AmSLaTeX}{%
 $\mathcal A$\lower.4ex\hbox{$\!\mathcal M\!$}$\mathcal S$-\LaTeX}
\def\BibTeX{{\rmfamily B\kern-.05em
 \textsc{i\kern-.025em b}\kern-.08em
  T\kern-.1667em\lower.7ex\hbox{E}\kern-.125emX}}
\makeatletter
\def\tmpcite#1{\@ifundefined{b@#1}{\textbf{?}}{\csname b@#1\endcsname}}%
\makeatother

\field{A}
\vol{99}
\no{2}
\title[Monitoring Temporal Properties using Interval Analysis]
      {Monitoring Temporal Properties using Interval Analysis}
\titlenote{The preliminary version of this paper was presented at the 8th International Workshop on Numerical Software Verification (NSV'15).} 
\authorlist{%
 \authorentry{Daisuke ISHII}{m}{labelA}
 \authorentry{Naoki YONEZAKI}{m}{labelA}
 \authorentry{Alexandre GOLDSZTEJN}{n}{labelB}
}
\affiliate[labelA]{The authors are with Tokyo Institute of Technology.}
\affiliate[labelB]{The author is with CNRS/IRCCyN.} 

\received{2015}{8}{11}
\revised{2015}{10}{5}

\begin{document}
\maketitle

\begin{summary}
Verification of temporal logic properties plays a crucial role in proving the desired behaviors of continuous systems. 
In this paper, we propose an interval method that verifies the properties described by a bounded signal temporal logic.
We relax the problem so that if the verification process cannot succeed at the prescribed precision, it outputs an inconclusive result. The problem is solved by an efficient and rigorous monitoring algorithm.
This algorithm performs a forward simulation of a continuous-time dynamical system, detects a set of time intervals in which the atomic propositions hold, and validates the property by propagating the time intervals.
In each step, the continuous state at a certain time is enclosed by an interval vector that is proven to contain a unique solution.
We experimentally demonstrate the utility of the proposed method in formal analysis of nonlinear and complex continuous systems.
\end{summary}

\section{Introduction}\label{intro}

Reasoning with the temporal logic properties in continuous systems is a challenging and important task that combines computer science, numerical analysis, and control theory.
Various methods for the verification of continuous and hybrid systems with bounded temporal properties have been developed, e.g., 
\cite{Plaku2009,Nghiem2010,David2012,Zuliani2013}, 
enabling the falsification of various properties (e.g., safety, stability, and robustness) of large and complex systems. However, the state-of-the-art tools are based on numerical simulations whose numerical errors frequently yield qualitatively wrong results, which become problematic even in statistical evaluations.

Computing rigorously approximated reachable sets is a fundamental process in formal methods for continuous systems.
Techniques based on interval analysis (Section~\ref{s:interval}) have proven practical in the reachability analysis of nonlinear and complex continuous systems~\cite{Eggers2008,Collins2008,Ramdani2011,Ishii2011,Chen2012,Gao2013:SMODE}. 
In these frameworks, the computation is \emph{$\delta$-complete} \cite{Gao2012}: assuming that function values may be perturbed within a predefined $\delta \in \SPosRatSet$, 
many generically undecidable problems become decidable.
However, $\delta$-complete verification of generic properties other than reachability is a challenging topic.

The contribution of this paper is to propose an interval method that verifies (bounded portions of) the signal temporal logic (STL) properties (Section~\ref{s:stl}) of a class of continuous-time dynamical systems (Section~\ref{s:cs}; extension to hybrid systems is straightforward).
Our method reliably computes three values: $\Valid$, $\Unsat$, and $\Unknown$. The method outputs $\Valid$ or $\Unsat$ when the soundness is guaranteed by interval analysis;
otherwise, when the verification fails after reaching a prescribed precision threshold, it outputs $\Unknown$.
Our method is based on validated interval analysis, and therefore it is reliable compared to the existing simulation-based monitoring tools, e.g., \cite{Maler2003,Fainekos2006a,Donze2010,Donze2010a}.
We show that simulation with numerical errors may compute an incorrect signal for a chaotic system.
In contrast with the existing tools that monitor a single behavior of a system,
our method monitors a set of possible behaviors using an interval-based technique;
therefore, the method can check the validity of the system.
In this sense, our approach can be viewed as an integration of reachability analysis and simulation-based monitoring methods.
The relaxation allowing $\Unknown$ results enables us to generate an efficient monitor for STL properties that can be regarded as a variant of $\delta$-complete procedures.
We demonstrate efficient and reliable monitors for several continuous systems including a chaotic system.

In Section~\ref{s:method}, we present an algorithm for monitoring STL properties based on the forward simulation that encloses a signal with a set of \emph{boxes} (i.e., interval vectors). 
For each atomic proposition involved in a property $\phi$ to be verified, the algorithm obtains an inner and outer approximation of the time intervals in which the proposition holds. The interval Newton operator is used for the purposes of accelerating the search of instants where the satisfaction of propositions changes, and certifying the uniqueness of these event within their enclosures, eventually certifying the sequence of consistent/inconsistent time intervals over time for each proposition.
Next, it modifies the set of time intervals according to the syntax of the property $\phi$; finally, it checks that $\phi$ holds at the initial time.
Using our implementation, we show that several benchmarks are verified efficiently, yet non-robust instances with respect to numerical errors are rejected (Section~\ref{s:ex}).
The implementation reliably analyzes a set of signals and provides a foundation for verification and parameter synthesis of complex systems.

\section{Related Work} \label{s:related}

Many previous studies have applied interval methods to reachability analyses of continuous and hybrid systems~\cite{Eggers2008,Collins2008,Ramdani2011,Ishii2011,Chen2012,Gao2013:SMODE}.
These methods output an over-approximation of reachable states as a set of boxes. 
%
Interval analysis often proves the unique existence of a solution within a resulting interval, and
it is also applicable to interval-based reachability analysis~\cite{Ishii2011,Goubault2014}.
Our method utilizes the proof in the verification of more generic temporal properties.


Reasoning of real-time temporal logic has been a research topic of interest~\cite{Alur1996,Shultz1997}. 
Numerical method for \emph{falsification} of a temporal property is straightforward~\cite{Maler2003}.
The algorithm simulates a signal of a bounded length and checks the satisfiability of the negation of the property described by a bounded temporal logic.
This paper presents an interval extension of this falsification method.

To falsify realistic nonlinear models efficiently, researchers have proposed a tree-search method~\cite{Plaku2009}, a Monte-Carlo optimization method~\cite{Nghiem2010}, and statistical model checking methods~\cite{David2012,Zuliani2013}.
Despite their successes, these methods are compromised by numerical error. To improve the reliability and practicality of the falsification, integration with our interval method will be a promising future direction.
An integrated statistical and interval method was also proposed in \cite{Wang2014} for reachability analysis.

To facilitate simulation-based verification of temporal properties, the \emph{robustness} concept has been proposed~\cite{Fainekos2006a,Donze2010,Nghiem2010}.
In these works, the degree of robustness defines the distance between a signal and a region over which a proposition holds.
If the absolute value of the degree is small, it is likely to be unreliable because of numerical errors.
Our method rigorously ensures robustness by verifying that every intersection between a signal and each boundary in the state space is enclosed with an interval.

There exist several methods for model checking of temporal logic properties~\cite{Podelski2006,Cimatti2014}.
\cite{Podelski2006} proposed a method specialized in stability properties, which is described as a specific form of temporal logic formula.
\cite{Cimatti2014} proposed a method that translates a verification problem into a reachability problem with the $k$-Liveness scheme,
which is incomplete in general settings.
Our method can be viewed as a bounded model checking method that validates a bounded temporal property when the property is satisfied by all signals emerging from the interval parameter value.


\section{Interval Analysis}
\label{s:interval}

This section introduces selected topics and techniques based on interval analysis~\cite{Moore1966,Neumaier1990}.


A (bounded) \textit{interval} $\a = [\LB{a}, \UB{a}]$ is a connected set of real numbers $\{b \in \RealSet ~|~ \LB{a} \leq b \leq \UB{a}\}$.
$\IntSet$ denotes the set of intervals. 
$\PosIntSet$ denotes the subset $\{[\LB{a},\UB{a}]\in\IntSet ~|~ \LB{a}\leq 0\}$.
For an interval $\a$,
$\LB{a}$ and $\UB{a}$ denote the lower and upper bounds, respectively; and
$\Inter{\a}$ denotes the interior %
$\{b \in \RealSet ~|~ \LB{a} < b < \UB{a}\}$.
%
$[a]$ denotes a point interval $[a,a]$.
The hypermetric between two intervals $\a$ and $\b$,
$\Dist{\a}{\b}$ is given by 
$\max ( |\UB{a}-\UB{b}| , |\LB{a}-\LB{b}| )$.
%
For a set $S \subset \RealSet$, $\Box S$ denotes the interval $[\inf S, \sup S]$.
All these definitions are naturally extended to interval vectors;
an $n$-dimensional \textit{box} (or interval vector) $\a$ is a tuple of $n$ intervals
$(\a_1, \ldots, \a_n)$, and
$\IntSet^n$ denotes the set of $n$-dimensional boxes.
For $a\in\RealSet^n$ and $\a\in\IntSet^n$,
we use the notation $a \in \a$, which is interpreted as $\ForAll{i}{\{1,\ldots,n\}} ~ a_i \in \a_i$.

In actual implementations, the interval bounds should be machine-representable floating-point numbers, and other real values are rounded in the appropriate directions.


Given a function $f : \RealSet^n \to \RealSet$,
$\f : \IntSet^n \to \IntSet$ is called an \emph{interval extension} of $f$ if and only if it satisfies the containment condition
$\ForAll{\a}{\IntSet^n} ~ \ForAll{a}{\a} ~ ( f(a) \in \f(\a) )$.
This definition is generalizable to function vectors $\f : \RealSet^n \to \RealSet^m$.
%
Given two intervals $\a,\b\in\IntSet$, we can compute interval extensions of the four operators $\circ \in \{+,-,\ast,/\}$ as $\Box\{\LB{a}\circ\LB{b}, \LB{a}\circ\UB{b}, \UB{a}\circ\LB{b}, \UB{a}\circ\UB{b}\}$ (assuming $0 \not\in \b$ for division).

For arbitrary intervals $\a,\b,\d\in\IntSet$, the \emph{extended division} 
$\Box\{d\in\d ~|~ \Exists{a}{\a}~\Exists{b}{\b}~ a = b d\}$
can be implemented as follows (see Section~4.3 of \cite{Neumaier1990}):
\[
	\ExtDiv(\a,\b,\d) := 
	\begin{cases}
		(\a/\b) \ \cap\ \d & \text{if}~ 0 \not\in \b\\
		\Box(\d \setminus (\LB{a}/\LB{b}, \LB{a}/\UB{b})) & \text{if}~ \LB{a} \!>\! 0 \!\in\! \b\\
		\Box(\d \setminus (\UB{a}/\UB{b}, \UB{a}/\LB{b})) & \text{if}~ \UB{a} \!<\! 0 \!\in\! \b\\
		\d & \text{if}~ 0 \in \a,\b
	\end{cases}
\]
In the second and third cases, when $\LB{b}=0$ (resp. $\UB{b}=0$), we set $\LB{a}/\LB{b}$ and $\UB{a}/\LB{b}$ as $-\infty$ and $\infty$ (resp. $\LB{a}/\UB{b}$ and $\UB{a}/\UB{b}$ as $\infty$ and $-\infty$).


Given a differentiable function $f(a) : \RealSet \to \RealSet$ and a domain interval $\a$,
a root $\tilde{a} \in \a$ of $f$ such that $f(\tilde{a}) = 0$ is included in the result of the \emph{interval Newton operator}
\[
	\hat{a} + \ExtDiv(-\f(\hat{a}), \f'(\a), \a-\hat{a}) \approx
	\left( \hat{a} - \frac{\f(\hat{a})}{\f'(\a)} \right) \cap \a,
\]
where $\hat{a} \in \a$, and $\f$ and $\f'$ are interval extensions of $f$ and the derivative of $f$, respectively. 
The first expression is always valid while the second expression is valid only when $\f'(\a)$ does not contain 0.
Iterative applications of the operator will converge.
Let $\a'$ be the result of applying the operator to $\a$. If $\a' \subseteq \Inter{\a}$, a unique root exists in $\a'$.

\section{Continuous-Time Dynamical Systems}
\label{s:cs}

We consider dynamical systems whose behaviors are described by ordinary differential equations (ODEs).

\begin{definition}
A \emph{continuous-time dynamical system} is a tuple 
$\CS :=\bigl( (u,x), U\Times X, X_\Init, \Flow \bigr)$
consisting of the following components:
\begin{itemize}
\item A vector of real-valued \emph{parameters} $u = (u_1,\ldots,u_m)$.
\item A vector of real-valued \emph{variables} $x = (x_1,\ldots,x_n)$.
\item A \emph{domain} $U \Times X \subseteq \RealSet^{m+n}$ for the valuation of the parameters and variables.
\item An \emph{initial domain} $X_\Init \subseteq X$.
\item A \emph{vector fields} $\Flow : U \Times X \to \RealSet^n$ (assuming Lipschitz continuity).
  %
\end{itemize}
\end{definition}
In this work, we specify domains $U$ and $X$ as boxes.
The behaviors of a system $\CS$ are formalized as \emph{signals}.
\begin{definition}
	Given a time interval $\t = [0,\UB{t}] \in \IntSet$ and a parameter value $\tilde{u} \in U$,
	a \emph{signal} of a continuous-time dynamical system $\CS$ is a function $\tilde{x} : \t \to X$ such that
    \begin{gather*}
		\tilde{x}(0) \in X_\Init \land 
		\ForAll{\tilde{t}}{\t} ~ \tfrac{d}{dt}\tilde{x}(\tilde{t}) = F(\tilde{u}, \tilde{x}(\tilde{t})).
    \end{gather*}
\end{definition}
$\Sigs_{\UB{t}}(\CS)$ denotes the set of signals of $\CS$ of length $\UB{t}$.

\begin{example} \label{ex:rotation}
	An anticlockwise rotation of a 2D particle can be modeled as a continuous-time dynamical system:
	\begin{align*}
		u &:= (u_1), \quad U := ([-0.1,0.1]), \\
		x &:= (x_1,x_2), \quad X := [-10,10]^2, \\
		X_\Init &:= \{(1, 0)\}, \\
		F(u,x) &:= \begin{pmatrix}
		u_1 & -1 \\
		1 & u_1
		\end{pmatrix}
		\begin{pmatrix}
			x_1 \\ x_2
		\end{pmatrix}.
	\end{align*}
A signal of this example is illustrated in Figure~\ref{f:rotation}.
The signal moves on the circle of radius 1 when $u_1 = 0$; the system is stable when $u_1 \leq 0$ and is unstable when $u_1 > 0$.
\end{example}

\begin{figure}[th]
\centering
\vspace{-1em}
\includegraphics[width=\linewidth]{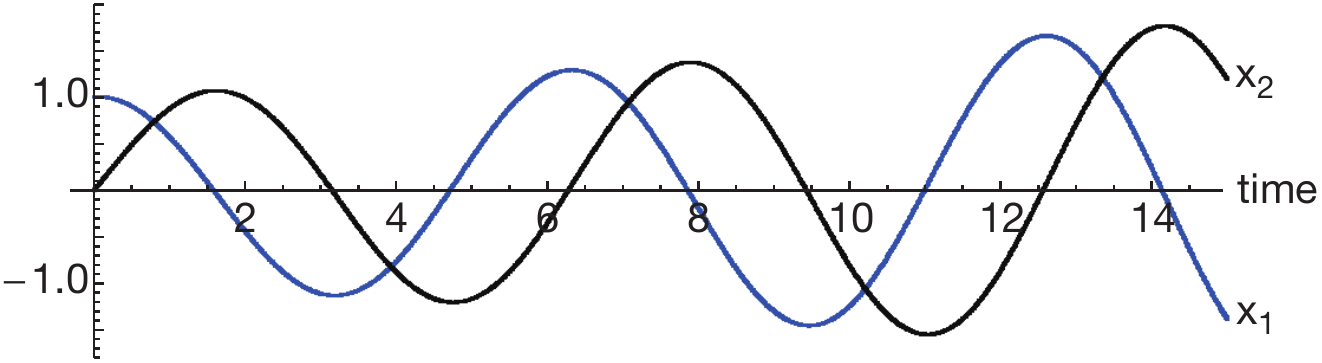}
\caption{A signal of the rotation system}
\label{f:rotation}
\end{figure}

\begin{example} \label{ex:lorenz}
	A well-known chaotic dynamical system is the Lorenz equation:
	\begin{align*}
		u &:= (u_1,u_2,u_3), ~ U := (10,28,2.5)+[-1,1]^3, \\
		x &:= (x_1,x_2,x_3), \quad X := [-50,50]^3, \\
		X_\Init &:= \{(15, 15, 36)\}, \\
		F(u,x) &:= 
		\begin{pmatrix}
			u_1(x_2-x_1) \\  x_1(u_2-x_3) - x_2 \\ x_1 x_2 - u_3 x_3
		\end{pmatrix}.
	\end{align*}
A signal of this system is illustrated in the upper part of Figure~\ref{f:lorenz}.
\end{example}

\begin{figure}[th]
\centering
\includegraphics[width=\linewidth]{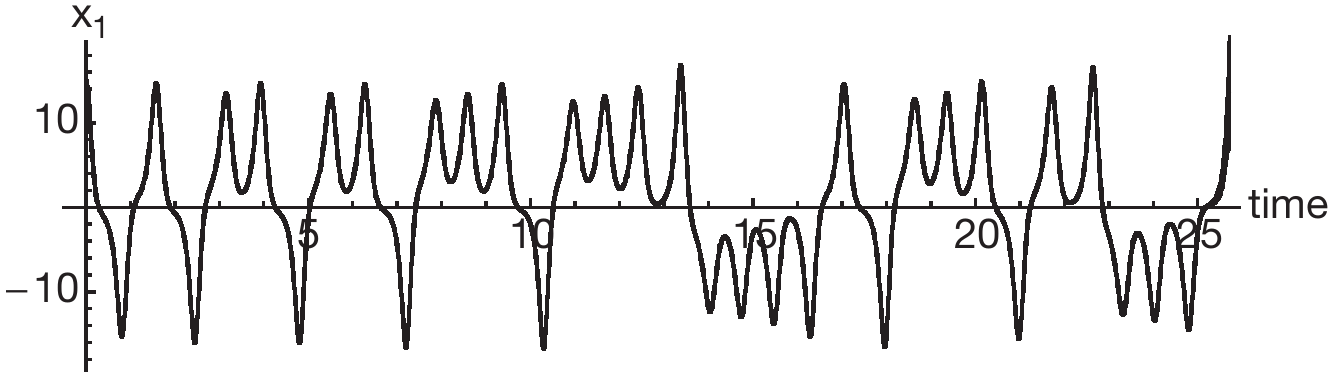}
\includegraphics[width=\linewidth]{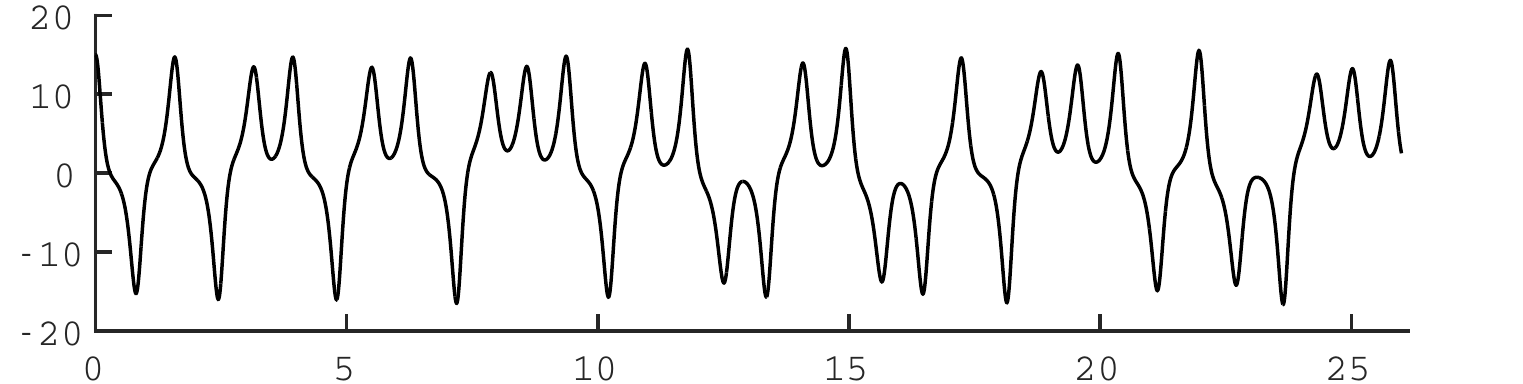}
\caption{Signals of the Lorenz system simulated by validated (upper) and non-validated (lower) numerical methods}
\label{f:lorenz}
\end{figure}

\subsection{ODE Integration using Interval Analysis}
\label{s:ode}

Using tools based on interval Taylor methods, such as CAPD\footnote{\url{http://capd.ii.uj.edu.pl/}} and VNODE~\cite{Nedialkov2006}, we can obtain an interval extension $\XC : \PosIntSet\to\IntSet^n$ of signals in $\Sigs_{\UB{t}}{\Struct{\CS}}$.
Given $\t \in {\PosIntSet}$, these tools perform stepwise integration of the flow function $F$ from the initial time $0$ to time $\UB{t}$, and output the value $\XC(\t)$.
At each step, interval Taylor methods verify the \emph{unique existence} of a solution in a box enclosure using the Picard-Lindel\"{o}f operator and Banach's fixpoint theorem.
Accordingly, when an interval enclosure $\XC(\t)$ is computed by an interval Taylor method, the following property holds:
\begin{gather*} \label{e:ode:unique}
	\ForAll{u}{U} ~ \ForAll{x_\Init}{X_\Init} ~ 
	\UExists{\tilde{x}}{\Sigs_{\UB{t}}(\CS)} ~ \\ \tilde{x}(0)=x_\Init \LAnd 
	\ForAll{\tilde{t}}{\t}~
	\tfrac{d}{dt} \tilde{x}(\tilde{t}) = F(u,\tilde{x}(\tilde{t})),
\end{gather*}
where $\exists!$ is interpreted as ``uniquely exists.''

In principle, if $F$ is Lipschitz continuous and we can assume arbitrary precision, we obtain an arbitrarily narrow interval enclosure $\XC([t])$ for $t \in \PosRealSet$.
However, because interval Taylor methods are implemented using machine-representable real numbers, they may fail to compute an enclosure when verifying the unique existence property, even at the smallest step size.

\begin{example} \label{ex:lorenz:sig}
	Signals of the Lorenz system in Example~\ref{ex:lorenz} (when $u := (10,28,2.5)$), computed with an interval method (CAPD) and a non-validated numerical method, are illustrated in Figure~\ref{f:lorenz}.
	Non-validated numerical methods may compute a wrong signal for a chaotic system as shown in this figure. 
	On the other hand, validated simulation of this system over a long period is difficult with double-precition floating-point numbers;
	the width of the interval enclosure computed by CAPD blows up after $25$ time units and the simulation fails.
\end{example}

\begin{figure*}[t]
\includegraphics[width=.9\linewidth]{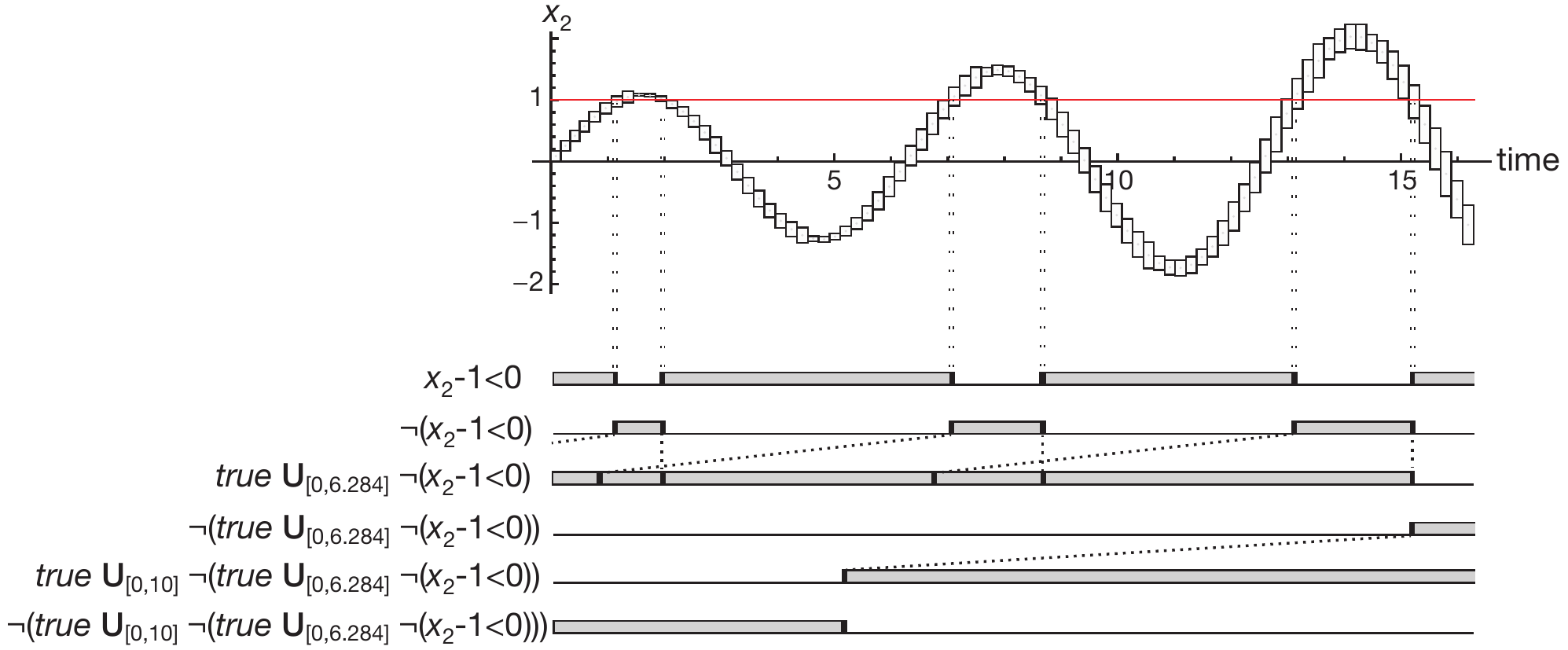}
\caption{Monitoring process on the rotation system}
\label{f:rotation:proc}
\end{figure*}

\section{Signal Temporal Logic}
\label{s:stl}

We consider a fragment~\cite{Maler2003} of the real-time metric temporal logic~\cite{Alur1996} whose temporal modalities are bounded by an interval $\t = [\LB{t},\UB{t}]$, where the bounds $\LB{t},\UB{t}$ are in $\PosRatSet$.
Following \cite{Maler2003}, we refer to this logic as the \emph{signal temporal logic} (STL).

\begin{definition}
We consider constraints in the real domain as atomic propositions.
The syntax of the STL formulae is defined by the grammar
\begin{align*}
	\phi ::=&~ \True ~|~ \Prop ~|~ \phi \lor \phi ~|~ \neg \phi ~|~ \phi\,\Until_{\ttt}\,\phi \\
	p ::=&~ f(x) < 0 
\end{align*}
where $\Prop$ belongs to a set of \emph{atomic propositions} $\APSet_\phi$, $\Until_{\ttt}$ is the ``until'' operator bounded by a non-empty positive time interval $\t \in {\PosIntSet}$, 
$x = (x_1,\ldots,x_n)$ is a vector of variables, and $f : \RealSet^{n} \to \RealSet$.
We use the standard abbreviations, e.g., $\phi_1\land\phi_2 := \neg(\neg\phi_1\lor\neg\phi_2)$, $\Eventually_\ttt \phi := \True\,\Until_\ttt\,\phi$ (``eventually''), and $\Always_\ttt\phi := \neg\Eventually_\ttt\neg\phi$ (``always'').
%
\end{definition}

\subsection{Semantics}

The necessary length $\Norm{\phi}$ of the signals for checking an STL formula $\phi$ is inductively defined by the structure of the formula:
\begin{align*}
	\Norm{p} &:= 0, & 
	\Norm{\phi_1 \lor \phi_2} &:= \Max\,(\Norm{\phi_1},\Norm{\phi_2}), \\
	\Norm{\neg \phi} &:= \Norm{\phi}, &
	\Norm{\phi_1 \,\Until_{\ttt}\, \phi_2} &:= \Max\,(\Norm{\phi_1},\Norm{\phi_2}) + \UB{t}.
\end{align*}
The map $\Obs : \APSet_\phi \to \PWS{X}$ associates each proposition $\Prop \in \APSet_\phi$ to a set $\Obs(\Prop) = \{x \!\in\! X ~|~ p(x)\}$.

When we check the satisfiability of $\phi$ at time $t$, we should have a signal of length $t_\Max := \Norm{\phi}+t$ (this value of $t_\Max$ is used in evaluating all subformulae of $\phi$).
Let $\tilde{x} \in \Sigs_{t_\Max}(\CS)$ and $\phi$ be an STL property.
Then, we have a satisfaction relation defined as follows:
\begin{align*}
	\Sig,t &\models \True &&\\
	\Sig,t &\models \Prop && \text{iff}~~ \Sig(t) \in \Obs(\Prop)\\
	\Sig,t &\models \phi_1 \lor \phi_2 && \text{iff}~~ \Sig,t \models \phi_1 \LOr \Sig,t \models \phi_2\\
	\Sig,t &\models \neg \phi && \text{iff}~~ \Sig,t \not\models \phi\\
	\Sig,t &\models \phi_1 \,\Until_\ttt\, \phi_2 \\
	\omit\rlap{~\text{iff $\Exists{t'}{(t+\t)} ~ \Sig,t' \models \phi_2 \LAnd 
	(\ForAll{t''}{[t,t']} ~ \Sig,t'' \models \phi_1)$}}
\end{align*}
At a given time $t$, $\phi_1\,\Until_\ttt\,\phi_2$ intuitively means that $\phi_2$ holds within the time interval $t\!+\!\t$ and that $\phi_1$ always hold until then.
We also have the following validation relation:
\[
	\CS \models \phi ~~\text{iff}~~ \ForAll{\Sig}{\Sigs_{\Norm{\phi}}(\CS)} ~ \Sig,0 \models \phi
\]

\subsection{Method for Monitoring STL Formulae}
\label{s:stl:monitoring}

Our interval method is based on the monitoring method proposed in \cite{Maler2003}, which decides whether a signal satisfies an STL property based on the numerical simulation of signals of bounded lengths.
This section explains this basic method.
First, we introduce the notion of consistent time intervals in the STL evaluation.
\begin{definition} \label{th:consistent}
	Let $\Sig$ be a signal of length $t_\Max$ and $\phi$ be an STL formula.
	We say that a left-closed and right-open interval $[\LB{t},\UB{t}) \subseteq \PosRealSet$ is \emph{consistent} with $\phi$ iff $\ForAll{t}{(\LB{t},\UB{t})}~ \Sig,t \models \phi$.%
		\footnote{The original definition~\cite{Maler2003} involves left-closed right-open time intervals $[\LB{t},\UB{t})$ so that they do not overlap and they can cover $[0,t_\Max]$. However, $\tilde{x}(t)>1 \ \equiv\ 1-\tilde{x}(t)<0$, with $\tilde{x}(t) := t$, is not true in the left-closed right-open interval $[1,2)$. In this paper, we only enforce the predicate to be true in the interior of time intervals $(\LB{t},\UB{t})$ to regard $[1,2)$ consistent. This has no impact on the soundness nor efficiency of the proposed method, since such bounds will be approximated by enclosing intervals in Definition~\ref{d:acti}.}
\end{definition}

Next, whether a signal satisfies property $\phi$ is checked as follows:
\begin{enumerate}
	\item For each atomic proposition $p$ in $\APSet_\phi$, monitor the signal of length $\Norm{\phi}$ and identify a non-overlapping set of consistent time intervals $T_p = \{\t_1,\ldots,\t_{n_p}\}$.
\item Following the parse tree of $\phi$ in a bottom-up fashion, obtain a set of consistent time intervals of $\phi$. For each construct of STL, obtain the set that is consistent with the sub-formula as follows:
\begin{align}
	T_{\neg \phi} &:= 
	\PosRealSet \setminus T_\phi \label{e:op:neg} \\
	T_{\phi_1\lor \phi_2} &:= T_{\phi_1} \cup T_{\phi_2} \label{e:op:lor} \\
	T_{\phi_1 \Until_{\ttt} \phi_2} &:= \nonumber \\
	\omit\rlap{\quad $\{\AlgName{Shift}_{\ttt}(\t_1 \cap \t_2) \cap \t_1 ~|~
	\t_1 \in T_{\phi_1}, \t_2 \in T_{\phi_2} \}$} \label{e:op:until}
\end{align}
where $\AlgName{Shift}_{\ttt}(\s) := [\LB{s}-\UB{t}, \UB{s}-\LB{t}) \cap \PosRealSet$.
\item Check whether $\Min\ T_\phi$ contains time 0.
	If yes, $\phi$ is satisfied; otherwise, it is not satisfied.
\end{enumerate}

\begin{example} \label{ex:stl}
	We consider the property 
	\begin{multline*}
		\Always_{[0,10]} \Eventually_{[0,6.284]}\, \neg (x_2\!-\!1 \!<\! 0) ~\equiv~\\
		\neg( \True \,\Until_{[0,10]}\, \neg ( \True \,\Until_{[0,6.284]}\, \neg(x_2\!-\!1 \!<\! 0) ))
	\end{multline*}
	for the model in Example~\ref{ex:rotation}, which describes that, within the initial $10$ time units, the signal $x_2$ increases beyond $1$ within every $6.284$ time units.
	Verification with the monitoring method (extended to an interval method) is illustrated in Figure~\ref{f:rotation:proc}, when the parameter is set as $u_1 := 0.1+[-10^{-3},10^{-3}]$.
\end{example}

\section{Interval-Based Monitoring Method}
\label{s:method}

In this section, we propose a reliable method for monitoring STL properties on continuous-time dynamical systems. This method is an interval extension of the monitoring method described in Section~\ref{s:stl:monitoring}.

Given a system $\CS$ and an STL property $\phi$, the proposed $\AlgName{MonitorSTL}$ algorithm (Algorithm~\ref{a:main}) outputs the following results:
$\Valid$ (implying that $\CS \models \phi$),
$\Unsat$ (implying that $\CS \models \neg\phi$), or
$\Unknown$ (meaning that the computation is inconclusive).
The algorithm implements the method described in Section~\ref{s:stl:monitoring}.
The sub-procedures $\AlgName{MonitorAP}$ (Section~\ref{s:method:map}) for monitoring atomic propositions, and $\Propagate$ and $\AlgName{ConsistentAtInitTime}$ (Section~\ref{s:method:propag}) for evaluating an STL formula are rendered rigorous and sound by interval analysis;
namely, the precision of every numerical computation is guaranteed, and the correctness of the monitoring method is assured by verifying the unique existence of a solution within its resulting interval.
Any errors introduced by the sub-procedures are captured by the catch clause at Line~5.

\begin{algorithm}[thb]
\caption{\label{a:main} $\AlgName{MonitorSTL}$ algorithm}

\begin{algorithmic}[1]
  \REQUIRE $\CS$, $\phi$
  \ENSURE $\Valid$, $\Unsat$, or $\Unknown$

  \STATE \textbf{try}
  \STATE \quad $\mathcal{T} := \AlgName{MonitorAP}(\CS, \phi)$
  \hfill\COMMENT{Step~1; Section~\ref{s:method:map}}
  \STATE \quad $\T_\phi := \Propagate(\mathcal{T}, \phi))$
  \hfill\COMMENT{Step~2; Section~\ref{s:method:propag}}
  \STATE \quad \textbf{return} {$\AlgName{ConsistentAtInitTime}(\T_\phi)$}
  \hfill\COMMENT{Step~3}
  \STATE \textbf{catch error return} $\Unknown$ \textbf{end try}
\end{algorithmic}
\end{algorithm}

Despite its efficient computational cost (Section~\ref{s:complexity}), the proposed method has some limitations.
First, the method is incomplete; it allows inconclusive computations and outputs $\Unknown$ when the interval computation is too imprecise to separate several solutions within an interval.
In practice, the $\Unknown$ output is valuable, because a numerically non-robust signal is rejected as an error in the verification process.
Second, although the algorithm validates system properties in principle, its success is guaranteed only over sufficiently small domains $U$ and $X_0$, particularly when evaluating nonlinear systems.
Third, the method is a bounded model-checking method, in the sense that the domain $U\Times X$ and the lengths of the signals are both bounded.

Our method is targeted at (i) a more generic framework, in which the possible initial and parameter values can be exhaustively enumerated, and (ii) statistical methods that treat the parameters as random variables and evaluate probabilistic STL properties.

\subsection{Approximation of Consistent Time Intervals}
\label{s:method:approx}

In this section, we introduce an interval approximation for the consistent time intervals (Definition~\ref{th:consistent}).
The basic idea is to enclose each bound of the consistent time intervals within a closed interval.
\begin{definition} \label{d:acti}
	Given a consistent time interval $\t = [\LB{t},\UB{t}) \subseteq [0,\Norm{\phi})$ that is consistent with an STL property $\phi$, we define an \emph{(interval) approximation} as a pair $(\s,\s')$ such that $\s,\s' \in\IntSet$, $\LB{t} \in \s$, and $\UB{t} \in \s'$.
\end{definition}
Given an approximation $(\s, \s')$ and a continuous signal $\Sig$, we have $\ForAll{t}{[\UB{s},\LB{s}')}~\Sig,t\models\phi$.

We now approximate a set of consistent time intervals $\{ \t_1,\ldots, \t_{n_\phi} \}$ as a set (or sequence) of approximations.
Instead of the set of pairs $\{ (\s_1, \s'_1), \AB \ldots, (\s_{n_\phi}, \s'_{n_\phi}) \}$,
we represent a set of approximations with the set $\{ (\s_1,\True), (\s'_1,\False), \ldots, \AB (\s_{n_\phi},\True), (\s'_{n_\phi},\False) \}$, where the tags $\True$ and $\False$ represent whether an element corresponds to a lower or an upper bound.
A set of approximations is interpreted as both \emph{outer} and \emph{inner approximations};
that is, each consistent time interval $\t_i$ is enclosed by the outer approximation $[\LB{s}_i,\UB{s}_i']$, and the inner approximation $(\UB{s}_i,\LB{s}'_i)$ is contained in  $\t_i$.
\begin{definition} \label{d:approxs}
	Consider a set $\T = \{\AB (\s_1,b_1), \AB \ldots, \AB (\s_{\#\mathbm{T}}, b_{\#\mathbm{T}})\}$ where $\s_i \in \IntSet$, $b_i \in \{\True,\False\}$, and $\#\T\in\NatSet$.
	The second element of each pair is a \emph{polarity value} that represents whether the pair is an enclosure of a lower or upper bound of a consistent time interval.
	We say that $\T$ is \emph{canonical} iff
	\begin{itemize}
		\item the elements can be sorted, i.e.,\\ $\ForAll{i}{\{1,\ldots,\#\T\!-\!1\}}~ \UB{s}_i < \LB{s}_{i+1}$,
		\item $\ForAll{i}{\{1,\ldots,\#\T\}}~ 0 \leq \UB{s}_i$, 
		\item $\ForAll{i}{\{1,\ldots,\#\T\!-\!1\}}~ b_i \neq b_{i+1}$, and
		\item $b_1 = \True$.
	\end{itemize}
	%
	We say that $\T$ is an \emph{(interval) approximation} of $T_\phi = \{\t_1,\ldots, \t_{n_\phi}\}$
	iff
	\begin{itemize}
	\item $\T$ is canonical,
	\item $\ForAll{\t}{T_\phi}~ \Exists{(\s,\True)}{\T}~ \LB{t} \in \s$, 
	\item $\ForAll{\t}{T_\phi}~ \UB{t}<t_\Max \Rightarrow \Exists{(\s,\False)}{\T}~ \UB{t} \in \s$, 
	\item $\ForAll{(\s,\True)} {\T}~ \UExists{\t}{T_\phi}~ \LB{t} \in \s$, and
	\item $\ForAll{(\s,\False)}{\T}~ \UExists{\t}{T_\phi}~ \UB{t} \in \s$.
	\end{itemize}
	%
\end{definition}
Given a set of consistent time intervals, its canonical approximation is a disjoint sequence of lower and upper bound enclosures; the sequence starts with a lower bound enclosure and ends with either a lower or an upper bound enclosure.
For $\t\in T_\phi$ such that $\UB{t} > \Norm{\phi}$, $\T_\phi$ may contain only its lower-bound enclosure.
$\Universe := \{([0],\True)\}$ and $\T_\False := \emptyset$ are the approximations of $T_\True=\PosRealSet$ and $T_\False=\emptyset$, respectively.

\begin{example} \label{ex:approx}
	Let $\CS$ be $(x, [0,10], 0, F(x) = 1)$ such that
	the variable $x$ represents the signal $\tilde{x}(t) = t$.
	Consider a property
		$\phi := \Eventually_{[0,2\pi]} (\cos x<0 \LAnd \sin x<0)$.
	$\Norm{\phi}$ is $2\pi$, and 
	the set of consistent time intervals within $[0,2\pi]$ is
	$T_{\cos x<0} := \{ [\frac{\pi}{2},\frac{3}{2} \pi) \}$, 
	$T_{\sin x<0} := \{ [\pi,2\pi) \}$, and
	$T_\phi := \{ [0, \frac{3}{2} \pi) \}$, respectively.
	Then, their approximations are
	\begin{align*}
		\T_{\cos x<0} & := \{([1.57,1.58],\True), ([4.71,4.72],\False)\}, \\
		\T_{\sin x<0} & := \{([3.14,3.15],\True), ([6.28,6.29],\False)\}, \\
		\T_\phi &:= \{ ([0],\True), ([4.71,4.72],\False) \}.
	\end{align*}
\end{example}

\subsection{Monitoring Atomic Propositions}
\label{s:method:map}

This section describes the $\MonitorAP$ procedure (Algorithm~\ref{a:map}) that,
given a system $\CS$ and an STL property $\phi$, computes a set $\mathcal{T}$ containing an approximated set $\T_p$ of consistent time intervals for each $p \in \APSet_\phi$.

\begin{algorithm}[t]
\caption{$\MonitorAP$ algorithm}
\label{a:map}
\begin{algorithmic}[1]
  \REQUIRE $\CS$, $\phi$
  \ENSURE $\mathcal{T}$
  %
  %
  \STATE $\mathcal{T} = \emptyset$
  \FOR{$p = f(x) < 0 \in \APSet_\phi$}
  \STATE $b \Asn p(\XC(0))$
  \STATE \textbf{if} $b$ \textbf{then} $\T_p \Asn \{([0],b)\}$ \textbf{else} $\T_p \Asn \emptyset$
  \STATE $\t \Asn [0,\Norm{\phi}]$; \quad $b \Asn \neg b$
  \LOOP
  \STATE $\t \Asn \SearchZero(\XC, F, f, \t)$
  \STATE \textbf{if} $\t = \emptyset$ \textbf{then} \textbf{break} \textbf{end if}
  \STATE 
  $\T_p := \T_p \cup \{(\t,b)\}$; \quad
  $\t \Asn [\UB{t}, \Norm{\phi}]$; \quad $b \Asn \neg b$
  \ENDLOOP
  \STATE $\mathcal{T} \Asn \mathcal{T} \cup \{\T_p\}$
  \ENDFOR

  \vspace{.5em}
  \RETURN $\mathcal{T}$
\end{algorithmic}
\end{algorithm}

The outer loop enumerates each atomic proposition $p$ of the form $f(x) < 0$. 
Lines~3--4 compute the initial polarity by evaluating the proposition at time 0; the set $\T_p$ is initialized accordingly.
Note that $\XC$ represents a solving process for the signals in $\Sigs_{\bar{t}}(\CS)$ (see Section~\ref{s:ode}), which can be regarded as a function $\PosIntSet \to \IntSet^n$.
The inner loop searches for bounds at which $f$ changes sign.
Line~7 invokes the $\AlgName{SearchZero}$ procedure (Algorithm~\ref{a:searchzero}), which searches for the \emph{earliest} bound at which $p$ switches consistency within the time interval $\t$, and outputs a sharp enclosure of the bound (or $\emptyset$ if there is no solution).
This result is stored in the set $\T_p$, and $\T_p$ is stored in $\mathcal{T}$.
%
%
%
%

\begin{example}
	For $\phi$ in Example~\ref{ex:approx}, $\MonitorAP$ computes
	$\mathcal{T}$ as $\{\T_{\cos x_1<0}, \AB \T_{\sin x_1<0}\}$.
\end{example}


The evaluation of atomic propositions $f(x) < 0$ switches between $\True$ and $\False$ at the root of $f : X\to\RealSet$.
%
The intersection between a signal $\Sig(t)$ and a boundary condition $f(x) = 0$ is searched by Algorithm~\ref{a:searchzero}.
As inputs, this algorithm accepts an interval extension of the signal $\XC : \PosIntSet\to\IntSet^n$, a vector field $F : X \to X$, the function $f$, and a time interval $\t_\mathrm{init} \in\PosIntSet$ to be searched.
The algorithm computes the time interval $\t \subseteq \t_\mathrm{init}$ that encloses the earliest root, i.e., 
\begin{multline} \label{e:zero:earliest}
	\t = \Box \bigl\{~ \Min \{t \!\in\! \t_\mathrm{init} ~|~ f(\Sig(t))=0\} \\
		~|~ \ForAll{\Sig}{\Sigs_{\bar{t}_\mathrm{init}}(\CS)} ~\bigr\}.
\end{multline}
$\SearchZero$ verifies that $\t$ encloses a unique bound, i.e.,
\begin{equation} \label{e:zero:unique}
	\ForAll{\Sig}{\Sigs_{\UB{t}_\mathrm{init}}(\CS)}~ \UExists{t}{\t}~ f(\Sig(t)) = 0.
\end{equation}
Alternatively, if no bound exists in $\t_\mathrm{init}$, $\SearchZero$ verifies the following:
\begin{equation} \label{e:zero:unsat}
	\ForAll{\Sig}{\Sigs_{\UB{t}_\mathrm{init}}(\CS)}~ \ForAll{t}{\t_\mathrm{init}}~ f(\Sig(t)) \neq 0.
\end{equation}

\begin{algorithm}[t]
\caption{\label{a:searchzero} $\SearchZero$ algorithm}

\begin{algorithmic}[1]
  \REQUIRE $\XC : \PosIntSet\to\IntSet^n$, 
     $F : X \to X$, 
	 $f : X \to \RealSet$, 
	 $\t_\mathrm{init} \in \PosIntSet$
	 \ENSURE $\t\in\IntSet$ 
  \PARAM 
  $\epsilon \in \SPosRatSet$, $\theta \in (0,1) \subset \RatSet$
  \STATE $\t \Asn \t_\mathrm{init}$ 
  \REPEATC{Lower bound reduction}
    \STATE $\t_\mathrm{bak} \Asn \t$
	\STATE $\d \Asn \Dt(f, \XC, F, \t)$
	\STATE $\t \Asn \LB{t}+\ExtDiv(-\f(\XC(\LB{t})),\ \d,\ \t-\LB{t})$
  \UNTIL{$d(\t_\mathrm{bak},\t) \leq \epsilon$}
  \STATE \textbf{if} $\t = \emptyset$ \textbf{then return} $\emptyset$ \textbf{end if}
  \vspace{.5em}
  \STATE $\t \Asn \LB{t}$; \quad $\Delta \Asn \infty$
  \LOOPC{Unique solution existence verification}
	\STATE $\d \Asn \Dt(f, \XC, F, \t)$
    \STATE \textbf{if} $\d \ni 0$ \textbf{then error end if}
    \STATE $\t' \Asn \LB{t}-{\f(\XC(\LB{t}))}/{\d}$
	\STATE \textbf{if} $\t' \subseteq \Inter{\t}$
	  \textbf{then} $\t \Asn \t'$; \textbf{break end if}
	\STATE $\Delta_\mathrm{bak} := \Delta$; \quad $\Delta := d(\t,\t')$
	\STATE $\t \Asn \t_\mathrm{init} \cap \AlgName{Inflate}(\t',1\!+\!\theta)$
	\STATE \textbf{if} $\Delta \geq (1\!-\!\theta)\, \Delta_\mathrm{bak}$ \textbf{then error end if}
  \ENDLOOP
  %
  %
  \vspace{.5em}
  \RETURN $\t$
\end{algorithmic}
\end{algorithm}

\begin{theorem}[Soundness]
	If $\SearchZero$ returns an interval $\t \neq \emptyset$, properties \eqref{e:zero:earliest} and \eqref{e:zero:unique} hold.
	If it returns $\emptyset$, property \eqref{e:zero:unsat} holds.
\end{theorem}

To justify the soundness os $\SearchZero$, we describe some details of Algorithm~\ref{a:searchzero}. 
Lines~2--6 repeatedly filter the time interval $\t$ using the interval Newton operator.
Line~4 (and Line~10) invokes the $\Dt$ procedure, which is given a function $f$ and computes an interval enclosure of the derivative $\frac{d}{dt}f(\Sig(t))$ over $\t$ using the chain rule
\[
	\tfrac{d}{dt}f(\Sig(\t)) \!=\! \tfrac{d}{dx}f(\Sig(\t)) \cdot \tfrac{d}{dt}\Sig(\t)
	\subseteq \f'(\XC(\t)) \cdot \F(\XC(\t)).
\]
Next, at Line~5, the interval Newton operator is applied.
To handle the numerator interval $\d$ containing zero, we implement the interval Newton by the extended division described in Section~\ref{s:interval}.
Because we expand the interval Newton on the lower bound $\LB{t}$ and the extended division encloses the values in the $\t-\LB{t}$ domain, the resulting $\t$ is filtered its inconsistent portion without losing the solutions or being expanded.
%
If the interval Newton returns $\emptyset$, $\SearchZero$ also returns $\emptyset$ to signal the unsatisfiability (Line~7).

Because $\t$ may contain several solutions, 
Line~8 of the algorithm resets $\t$ to the lower bound as a starting value for computing the enclosure of the earliest solution.
Then, $\SearchZero$ checks that the time interval contains a unique solution.
To this end, it applies the interval Newton with the inclusion test to prove the unique existence of a solution within the contracted interval $\t'$ (Lines~9--17).
The interval Newton verification is repeated with an inflation process of the time interval (see \cite{Goldztejn2010:RC} for a detailed implementation).
If Line~18 is reached with no error, the time interval $\t$ is a sharp enclosure of the first zero of $f(\Sig(t))=0$. 

When $\SearchZero$ is implemented with machine-representable real numbers or when there is a tangency between the signal and the boundary condition, an error may result.
%
%
Line~11 of $\SearchZero$ outputs an error if the derivative on an (inflated) time interval contains zero.
At Line~16, we limit the number of iterations by specifying a threshold $1\!-\!\theta$ for the inflation ratio between two consecutive contraction amounts as in \cite{Goldztejn2010:RC}.


\subsection{Evaluation of STL Properties}
\label{s:method:propag}

We now describe the procedures for evaluating STL formulae at Lines~3 and 4 of Algorithm~\ref{a:main}.
Propagation of a set of monitored time intervals that are consistent with the atomic propositions is implemented as a rigorous and sound but incomplete procedure.

To evaluate the approximated sets, we extend the evaluation procedure on sets of consistent time intervals described in Section~\ref{s:stl:monitoring}.
Algorithm~\ref{a:propagate} implements Step~2 of the procedure, which propagates the STL formulae over a set of time intervals.

\begin{algorithm}[t]
\caption{$\Propagate$ algorithm}
\label{a:propagate}

\begin{algorithmic}[1]
  \REQUIRE $\phi$, $\mathcal{T} = \{\T_p\}_{p \in \APSet_\phi}$
  \ENSURE $\T_\phi$

  \STATE \textbf{switch} $\phi$
  \STATE \textbf{case} $p$ : 
  \STATE \qquad \textbf{return} {$\T_p$}
  \STATE \textbf{case} $\neg \phi'$ : 
  \STATE \qquad \textbf{return} $\AlgName{Invert}(\Propagate(\phi', \mathcal{T}))$
  \STATE \textbf{case} $\phi_1 \lor \phi_2$ : 
  \STATE \qquad $\T_{1} := \Propagate(\phi_1, \mathcal{T})$
  \STATE \qquad $\T_{2} := \Propagate(\phi_2, \mathcal{T})$
  \STATE \qquad \textbf{return} {$\AlgName{Join}(\T_{1}, \T_{2})$}
  \STATE \textbf{case} $\phi_1 \Until_\ttt \phi_2$ : 
  \STATE \qquad $\T_{1} := \Propagate(\phi_1, \mathcal{T})$
  \STATE \qquad $\T_{2} := \Propagate(\phi_2, \mathcal{T})$
  \STATE \qquad \textbf{return} $\AlgName{ShiftAll}_\ttt(\T_{1}, \T_{2})$
  \STATE \textbf{end switch}
\end{algorithmic}
\end{algorithm}

We now handle the approximated sets by extending the operations \eqref{e:op:neg}--\eqref{e:op:until} on sets of time intervals.
The procedures for the operations $\Invert$, $\JoinOp$, $\Intersect$, and $\ShiftAll$ are described in Figure~\ref{f:operations} in the appendix.
Note that, some operations cause ambiguities when handling non-canonical approximated time intervals.
Such a situation is exemplified below.
To avoid these ambiguities, our implementation results in an error once a resulting set becomes non-canonical.%
\footnote{
To output $\Unknown$ only due to the insufficient precision of numerical computation,
the procedure should branch the process and proceed evaluation for both cases;
implementation of such a procedure remains as a future work.}

\begin{example}
Consider the same timer system as in Example~\ref{ex:approx}, i.e. $\tilde{x}(t) := t$,
and the property 
$\phi := \Eventually_{[0,\UB{t}]} \neg (x-1 < 0 \lor 1-x < 0)$, where $\UB{t} \in \PosRealSetS$.
The subformula 
$x-1 < 0 \lor 1-x < 0$ is consistent at every time except at $t = 1$, therefore, the set of consistent time intervals is $T := \{[0,1),[1,t_\Max)\}$.
Assume $T$ is approximated with a non-canonical set $\{([0],\True),([0.95,1.1],\False),([0.9,1.05],\True)\}$.\footnote{The bound enclosures are usually very accurate, but kept large on this example to emphasize their impact.}
To verify $\phi$, the procedures $\Invert$ and $\AlgName{ShiftAll}_{[0,\UB{t}]}$ should be applied.
However, as illustrated in Figure~\ref{f:overlap}, we cannot decide whether the overlapping boundary intervals should be removed and expanded, or separated, and $\Propagate$ results in an error.\footnote{The property $\Eventually_{[2,3]} \neg (x-1)^2 < 0$ is verified in the same way. The set $T$ is consistent with the atomic proposition $(x-1)^2 < 0$. The verification will result in an error when $\SearchZero$ computes an enclosure of $T$ at time 1.} This ambiguous situation is avoided by using only canonical approximations. Note that, in some cases, this local ambiguity does not impact the global consistency. In the case of this example, if $\UB{t} < 0.9$, both scenarios lead to $\False$, and forking the resolution process would be able to resolve the local ambiguity.

\begin{figure}[ht]
	\vspace{-1em}
\centering
\includegraphics[width=.8\linewidth]{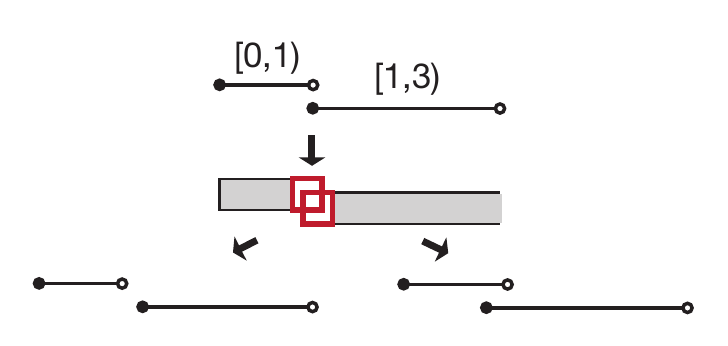}
\caption{Ambiguity caused by overlapping bounds}
\label{f:overlap}
\end{figure}
\end{example}

The following claims state that the procedures are closed in the canonical approximated sets, and the $\Propagate$ procedure is sound.

\begin{lemma} \label{th:canonical}
	Let $\T_1$ and $\T_2$ be canonical approximated sets.
	If $\T$ results from $\Invert(\T_1)$, $\JoinOp(\T_1,\T_2)$, $\Intersect(\T_1,\T_2)$, or $\ShiftAll(\T_1,\T_2)$, and if no $\mathbf{error}$ occurs in these procedures, then $\T$ is canonical.
\end{lemma}
\begin{proof}
	See Appendix~\ref{s:canonical:proof}.
\end{proof}

\begin{theorem}[Soundness] \label{th:soundness}
	Consider an STL formula $\phi$ and a set $\mathcal{T} = \{\T_p\}_{p \in \APSet_\phi}$ of approximated sets of time intervals that are consistent with atomic propositions.
	If $\T_\phi = \Propagate(\phi, \mathcal{T})$,
	then $\T_\phi$ is an approximation of $T_\phi$.
\end{theorem}
\begin{proof}
	See Appendix~\ref{s:soundness:proof}.
\end{proof}

Finally, we obtain $\T_\phi$ and conclude that $\phi$ is $\Valid$ if $\s_1$ is the smallest interval in $\T_\phi$ and $\UB{s}_1 \leq 0 < \LB{s}_1'$, $\Unsat$ if $\T_\phi = \emptyset$ or $0 < \LB{s}_1$, or $\Unknown$ if $0 \in [\LB{s}_1,\UB{s}_1)$.
The computation is performed by $\AlgName{ConsistentAtInitTime}$ (see Algorithm~\ref{a:inittime} in the appendix).

\subsection{Computational Cost}
\label{s:complexity}

The time complexity of $\Propagate$ is bounded by the product of the size (i.e., the number of operators) of the considered STL formula $\phi$ and the cost of the procedures $\Invert$, $\JoinOp$, and $\ShiftAll$ in Figure~\ref{f:operations}.
The complexity of the procedures on approximated sets is polynomial in the number of intersections of the signal and the atomic proposition bounds (see Appendix~\ref{s:complexity:op}). The number of iterations in $\MonitorAP$ is bounded by the product of the size of $\APSet_\phi$ and the maximum number of bounds detected for an atomic proposition, i.e., $\Max_{p\in\APSet_\phi}\ \#\T_p$.

The number of bounds detected depends on the oscillations of the ODE solution and the predicate bound. Although it can be very high in theory, it is usually quite small. In the generic case of non tangent intersection between the ODE solution and the predicate bound, the $\SearchZero$ procedure has a quadratic convergence and therefore a very low computational cost. The main computational cost of the method is therefore the validated simulation of the ODE. This cost is difficult to foresee: It highly depends on the ODE and the solver. For example, validated solvers are currently quite inefficient in solving stiff ODEs, and require to iterate many steps leading to a high computational cost. The complexity of $\Propagate$ is bounded by the product of the size (i.e., the number of operators) of the considered STL formula $\phi$ and the cost of the procedures in Figure~\ref{f:operations}.

\begin{table*}[t]
	\centering
	\caption{\label{t:ex:rotation} Experimental results (rotation)} 
    \begin{tabular}{l|c|c|c|r|r|r|r} \hline \hline
		$\phi$ & \!\#$\APSet_\phi$\! & $\GBnd$ & $\Wid{\s_1}$ &
		\!\#$\Valid$\! & \!\#$\Unsat$\! & \!\#$\Unknown$\! & time \\
		\hline
		\multirow{5}{*}{
		\hspace{-1.5em}
		\begin{tabular}{l}
		$\Always_{[0,\GBnd]} \Eventually_{[0,6.284]} \neg(x_2-1<0)$
		\end{tabular}}
		&   & \multirow{3}{*}{100} 
			  &  0               &  507 & 493 &   0+0 & 0.51s \\
		&   & &  $2\cdot10^{-6}$ &  483 & 508 &   9+0 & 0.52s \\
		& 1 & &  $2\cdot10^{-3}$ &    0 & 462 & 538+0 & -- \\
		\cline{3-8}
		&   & \multirow{2}{*}{10} 
			  &  0               &  490 & 510 &   0+0 & 0.03s \\
		&   & &  $2\cdot10^{-3}$ &  270 & 470 & 260+0 & 0.02s \\
		\hline
		\multirow{5}{*}{
		\hspace{-1.5em}
		\begin{tabular}{l}
		$\Always_{[0,\GBnd]} \Eventually_{[0,6.284]}$\\
		$(\neg(x_2-1<0) \land \Eventually_{[0,3.142]} \neg(-x_2-1<0))$
		\end{tabular}}
		&   & \multirow{3}{*}{100} 
			  &  0               &  485 & 515 &   0+0 & 1.1s \\
		&   & &  $2\cdot10^{-6}$ &  353 & 505 &  26+116 & 1.03s \\
		& 2 & &  $2\cdot10^{-3}$ &    0 & 463 & 537+0 & -- \\
		\cline{3-8}
		&   & \multirow{2}{*}{10} 
			  &  0               &  514 & 486 &   0+0 & 0.05s \\
		&   & &  $2\cdot10^{-3}$ &   86 & 493 & 421+0 & 0.06s \\
		\hline
		\multirow{5}{*}{
		\hspace{-1.5em}
		\begin{tabular}{l}
		$\Always_{[0,\GBnd]} \Eventually_{[0,6.284]} (\neg(x_2-1<0) \land $\\
		$\Eventually_{[0,1.571]} (\neg(-x_2<0) \land$\\
		$\Eventually_{[0,1.571]} (\neg(-x_2-1<0) \land$\\
		$\Eventually_{[0,1.571]} (-x_2<0)\ )))$
		\end{tabular}}
		&   & \multirow{3}{*}{100} 
			  & 0               &  482 & 518 &   0+0 & 1.7s \\
		&   & & $2\cdot10^{-6}$ &  346 & 498 & 18+138 & 1.6s \\
		& 3 & & $2\cdot10^{-3}$ &    0 & 0   & 1000+0 & -- \\
		\cline{3-8}
		&   & \multirow{2}{*}{10} 
			  & 0               &  516 & 484 &   0+0 & 0.08s \\
		&   & & $2\cdot10^{-3}$ &   84 &   0 & 916+0 & 0.09s \\
		\hline
		\multirow{5}{*}{
		\hspace{-1.5em}
		\begin{tabular}{l}
		$\Always_{[0,\GBnd]} \Eventually_{[0,6.284]} (\neg(x_2-1<0) \land$\\
		$\Eventually_{[0,0.786]} ((x_2\!-\!0.707\!<\!0) \land
		\Eventually_{[0,0.786]} (\neg(-x_2\!<\!0) \land$\\
		$\Eventually_{[0,0.786]} (\neg(-x_2\!-\!0.707\!<\!0) \land
		\Eventually_{[0,0.786]} (\neg(-x_2\!-\!1\!<\!0) \land$\\
		$\Eventually_{[0,0.786]} ((-x_2\!-\!0.707\!<\!0) \land
		\Eventually_{[0,0.786]} ((-x_2\!<\!0) \land$\\
		$\Eventually_{[0,0.786]} \neg(x_2-0.707<0)\ )))))))$
		\end{tabular}}
		&   & \multirow{3}{*}{100} 
			  & 0               &  490 & 510 &   0+0 & 2.7s \\
		&   & & $2\cdot10^{-6}$ &  352 & 477 &  74+97 & 2.7s \\
		& 5 & & $2\cdot10^{-3}$ &    0 &   0 & 1000+0 & -- \\[.1em]
		\cline{3-8}
		&   & \multirow{2}{*}{10} 
			  & 0               &  499 & 501 &   0+0 & 0.14\\
		&   & & $2\cdot10^{-3}$ &    0 &   0 & 1000+0 & -- \\[.1em]
		\hline
	\end{tabular}
\end{table*}

\section{Experiments}
\label{s:ex}

We have implemented the proposed method and experimented on two examples to confirm the effectiveness of the method.
Experiments were run on a 3.4GHz Intel Xeon processor with 16GB of RAM.

\subsection{Implementation}
\label{s:impl}

Algorithms~\ref{a:main}--\ref{a:inittime} were implemented in OCaml and C/C++.
ODEs were solved by procedures in the CAPD library.
%
The configurable parameters $t_\Min$, $\epsilon$, and $\theta$ correspond to the smallest integration step size that CAPD can take, the threshold used in Figure~\ref{a:searchzero}, and the threshold used in $\AlgName{Inflate}$, respectively.
In the experiments, these parameters were set as $t_\Min := 10^{-14}$, $\epsilon := 10^{-14}$, and $\theta := 0.01$.

\subsection{Verification of the Rotation System}


We verified the system in Example~\ref{ex:rotation} on four STL formulae.
The specifications and results of this experiment are summarized in  Table~\ref{t:ex:rotation}.
The first column lists the STL formulae in which the bound $\GBnd$ of each $\Always$ operator is parameterized and set to either $\GBnd:=100$ or $\GBnd:=10$.
The column ``\#$\APSet_\phi$'' represents the number of atomic propositions in each $\phi$.
In each verification, the parameter value $u_1$ was first randomly selected from $[-0.1,0.1]$ and then modified to $u_1 := u_1 + \u_1$, where $\u_1$ was any of $[0]$, $[-10^{-6},10^{-6}]$, or $[-10^{-3},10^{-3}]$.
The column ``$\Wid{\u_1}$'' indicates the interval used in each verification.

The considered STL properties are assumed to hold if $u_1 > 0$ and not to hold if $u_1 < 0$.
Each STL property was verified for 1000 times.
The columns ``\#$\Valid$'', ``\#$\Unsat$'', and ``\#$\Unknown$'' list the numbers of runs resulting in each output; the ``\#$\Unknown$'' outputs are separated with `$+$' according to whether it was caused by an error in the $\SearchZero$ algorithm or an error in the $\Propagate$ and $\AlgName{ConsistentAtInitTime}$ algorithms.
The column ``time'' lists the average CPU time taken for a $\Valid$ verification.

From the results, we can observe that the rates of inconclusive runs were related to the simulation lengths, the uncertainties in the parameter values, and the size of the formula $\phi$.
$\Unknown$ results were generated by the interval Newton process in $\SearchZero$ and the undecidable situations in $\Propagate$ and $\AlgName{ConsistentAtInitTime}$. In this experiment, verification failures increased as the value of $u_1$ approached 0 and the signal and boundary condition became close to tangent.
When the parameter values were exact and $\Wid{\u_1} = 0$, all the verifications succeeded even under near-singular conditions because the considered signals were always enclosed with tight intervals.
As coarser intervals were appended to the parameter values and the simulation lengths became longer, the number of $\Unknown$ results increased; meanwhile, the number of $\Valid$ results decreased more rapidly than the number of $\Unknown$ results because a $\Valid$ verification required detecting a number of bounds for each atomic proposition. Any detection failure resulted in $\Unknown$.

The bottleneck of the verification process is the $\SearchZero$ algorithm that integrates ODEs and searches for boundary intervals. The number of calls to $\SearchZero$ depends on the size of $\APSet_\phi$ and the number of bounds as described in Section~\ref{s:complexity}.
Therefore, the runtime increased linearly in either the number of atomic propositions or the simulation length that should be proportional to the number of bounds.
The cost of evaluation of the STL formulae seemed relatively small and not affecting the overall timings.

\subsection{Verification of the Lorenz System}

We verified the system in Example~\ref{ex:lorenz} on the following STL formula:
\begin{multline} \label{e:lorenz}
	\Always_{[0,15]} (\neg(-x_1-15<0) \Rightarrow \\
	\Eventually_{[0.5,5]} \Always_{[0,1]} ((x_1\!-\!10)^2 \!+\! (x_2\!-\!10)^2 \!-\! 150<0))
\end{multline}
In each verification, the parameters were set to exact values randomly selected from the domain.
The signal $(x_1, x_2)$ oscillates on either the positive or the negative side. According to the formula, when $x_1$ descends below $-15$, $(x_1,x_2)$ moves into the disk $(x_1-10)^2+(x_2-10)^2 <150$ after some duration in the interval $[0.5,5]$ and remains in the disk for at least 1 time unit.

The experimental results are summarized in Table~\ref{t:ex:lorenz}.
As in Table~\ref{t:ex:rotation}, the columns (from left to right) represent the number of atomic propositions, the numbers of $\Valid$, $\Unsat$, and $\Unknown$ verification results in 1000 runs, and the average CPU time for a $\Valid$ verification.

\begin{table}[ht]
\begin{center}
	\caption{\label{t:ex:lorenz} Experimental results (Lorenz)} 
    \begin{tabular}{c|r|r|r|r} \hline\hline
		\#$\APSet_\phi$ & \#$\Valid$ & \#$\Unsat$ & \#$\Unknown$ & time \\
		\hline
		2 & 566 & 413 & 21 & 9.2s \\
		\hline
	\end{tabular}
\end{center}
\end{table}

This experiment demonstrated that the proposed method can handle a chaotic system with a nonlinear atomic proposition. In such systems, non-validated numerical methods frequently output wrong results because of rounding errors, as shown in the next section.
As explained in Example~\ref{ex:lorenz}, CAPD integration generated a coarse enclosure of the signal (around $25.8$ time units), which introduced errors in the integration process $\XC$ or the interval Newton process.
These errors would account for the 21 $\Unknown$ results in Table~\ref{t:ex:lorenz}.

\subsection{Comparison with Breach Toolbox}

For comparative purposes, we ran the above problems on the Breach Toolbox~\cite{Donze2010a} (built from commit ed1178c in the Mercurial repository), a tool for STL verification based on numerical computation with rounding errors.
Breach can check the satisfiability and the \emph{robustness}, which is quantified by a positive or negative real value based on the distance between a considered signal and the bound in the state space where the satisfaction of the STL property switches.

For the rotation system, when the parameter value $u_1$ approached 0 (specifically, at $u_1 := 0.001$), Breach returned $\Unsat$, whereas our implementation returned $\Valid$.
This incorrect verification was implied by the low robustness value.
In this example, the robustness was low for all parameter values because the initial part of the signal was close to the bounds of the atomic propositions.

For the Lorenz system, the numerical integration process of Breach yielded incorrect signals, as explained in Example~\ref{ex:lorenz:sig}; therefore, the verification results were unreliable. 
For example, when $u = (10,28,2.5)$, Breach reported an $\Unsat$ verification of property~\eqref{e:lorenz}, whereas our method returned certified $\Valid$.

Breach ran more quickly than our implementation: it required less than 0.01s for both problems.

\section{Conclusions}

We have presented a sound STL validation method for checking that all initialized signals satisfy the properties of a system.
The proposed method detects a witness signal and verifies its unique existence using an interval-based ODE integration and an interval Newton method.
The experimental results demonstrate the potential for the method as a practical tool.

In future work, we will improve our method and implementation to handle hybrid systems and large and uncertain initial values.
Examples in a realistic setting should be demonstrated with the implementation.

\section*{Acknowledgments}

This work was partially funded by JSPS (KAKENHI 25880008 and 15K15968).

\bibliographystyle{ieicetr}
\bibliography{mitl}

\begin{thebibliography}{10}

\bibitem{Plaku2009}
E.~Plaku, L.E. Kavraki, and M.Y. Vardi, ``{Falsification of LTL Safety
  Properties in Hybrid Systems},'' TACAS, \LNCS{5505}, pp.368--382, 2009.

\bibitem{Nghiem2010}
T.~Nghiem, S.~Sankaranarayanan, G.~Fainekos, F.~Ivancic, A.~Gupta, and G.J.
  Pappas, ``{Monte-Carlo Techniques for Falsification of Temporal Properties of
  Non-Linear Hybrid Systems},'' HSCC, pp.211--220, 2010.

\bibitem{David2012}
A.~David, D.~Du, K.G. Larsen, A.~Legay, M.~Miku{\v{c}}ionis, D.B. Poulsen, and
  S.~Sedwards, ``{Statistical Model Checking for Stochastic Hybrid Systems},''
  Electronic Proceedings in Theoretical Computer Science, vol.92, pp.122--136,
  aug\ 2012.

\bibitem{Zuliani2013}
P.~Zuliani, A.~Platzer, and E.M. Clarke, ``{Bayesian statistical model checking
  with application to Stateflow/Simulink verification},'' Formal Methods in
  System Design, vol.43, no.2, pp.338--367, 2013.

\bibitem{Eggers2008}
A.~Eggers, M.~Franzle, and C.~Herde, ``{SAT Modulo ODE : A Direct SAT Approach
  to Hybrid Systems},'' ATVA, \LNCS{5311}, no.1, pp.171--185, 2008.

\bibitem{Collins2008}
P.~Collins and A.~Goldsztejn, ``{The Reach-and-Evolve Algorithm for
  Reachability Analysis of Nonlinear Dynamical Systems},'' Electronic Notes in
  Theoretical Computer Science, vol.223, no.639, pp.87--102, dec\ 2008.

\bibitem{Ramdani2011}
N.~Ramdani and N.S. Nedialkov, ``{Computing reachable sets for uncertain
  nonlinear hybrid systems using interval constraint-propagation techniques},''
  Nonlinear Analysis: Hybrid Systems, vol.5, no.2, pp.149--162, may\ 2011.

\bibitem{Ishii2011}
\IshiiEn, \UedaEn, and \HosobeEn, ``{An interval-based SAT modulo ODE solver
  for model checking nonlinear hybrid systems},'' International Journal on
  Software Tools for Technology Transfer (STTT), vol.13, no.5, pp.449--461,
  2011.

\bibitem{Chen2012}
X.~Chen, E.~Abraham, and S.~Sankaranarayanan, ``{Taylor Model Flowpipe
  Construction for Non-linear Hybrid Systems},'' IEEE Real-Time Systems
  Symposium, pp.183--192, 2012.

\bibitem{Gao2013:SMODE}
S.~Gao and E.M. Clarke, ``{Satisfiability Modulo ODEs},'' FMCAD, pp.105--112,
  2013.

\bibitem{Gao2012}
S.~Gao, J.~Avigad, and E.M. Clarke, ``{Delta-Decidability over the Reals},''
  LICS, pp.305--314, 2012.

\bibitem{Maler2003}
O.~Maler and D.~Nickovic, ``{Monitoring Temporal Properties of Continuous
  Signals},'' FORMATS, \LNCS{3253}, pp.152--166, 2004.

\bibitem{Fainekos2006a}
G.~Fainekos, A.~Girard, and G.~Pappas, ``{Temporal logic verification using
  simulation},'' FORMATS, \LNCS{4202}, vol.4202, pp.171--186, 2006.

\bibitem{Donze2010}
A.~Donz\'{e} and O.~Maler, ``{Robust Satisfaction of Temporal Logic over
  Real-Valued Signals},'' FORMATS, \LNCS{6246}, pp.92--106, 2010.

\bibitem{Donze2010a}
A.~Donz{\'{e}}, ``{Breach, a toolbox for verification and parameter synthesis
  of hybrid systems},'' CAV, \LNCS{6174}, pp.167--170, 2010.

\bibitem{Goubault2014}
E.~Goubault, O.~Mullier, and M.~Kieffer, ``{Inner Approximated Reachability
  Analysis},'' HSCC, pp.163--172, 2014.

\bibitem{Alur1996}
R.~Alur, T.~Feder, and T.A. Henzinger, ``{The Benefits of Relaxing
  Punctuality},'' Journal of the ACM, vol.43, no.1, pp.116--146, 1996.

\bibitem{Shultz1997}
B.~Shultz and B.J. Kuipers, ``{Proving properties of continuous systems :
  qualitative simulation and temporal logic},'' Artificial Intelligence,
  vol.92, no.96, pp.91--129, 1997.

\bibitem{Wang2014}
Q.~Wang, P.~Zuliani, S.~Kong, S.~Gao, and E.~Clarke, ``{SReach : Combining
  Statistical Tests and Bounded Model Checking for Nonlinear Hybrid Systems
  with Parametric Uncertainty},'' tech. rep., Carnegie Mellon University, 2014.

\bibitem{Podelski2006}
A.~Podelski and S.~Wagner, ``{Model Checking of Hybrid Systems : From
  Reachability towards Stability},'' HSCC, \LNCS{3927}, pp.507--521, 2006.

\bibitem{Cimatti2014}
A.~Cimatti, A.~Griggio, S.~Mover, and S.~Tonetta, ``{Verifying LTL Properties
  of Hybrid Systems with K-Liveness},'' CAV, \LNCS{8559}, pp.424--440, 2014.

\bibitem{Moore1966}
R.E. Moore, {Interval Analysis}, Prentice-Hall, 1966.

\bibitem{Neumaier1990}
A.~Neumaier, {Interval Methods for Systems of Equations}, Cambridge University
  Press, 1990.

\bibitem{Nedialkov2006}
N.S. Nedialkov, ``{VNODE-LP --- A Validated Solver for Initial Value Problems
  in Ordinary Differential Equations},'' tech. rep., McMaster University, 2006.

\bibitem{Goldztejn2010:RC}
A.~Goldsztejn and L.~Jaulin, ``{Inner approximation of the range of
  vector-valued functions},'' Reliable Computing, vol.14, pp.1--23, 2010.

\end{thebibliography}

\appendix*
\section{Omitted Procedures and Proofs}

Procedures of the operations on approximated sets are specified in Figure~\ref{f:operations}.
$\Invert$, $\JoinOp$, $\Intersect$, and $\ShiftAll$ implement the operations in Step~2 of Section~\ref{s:stl:monitoring} as procedures that modify the set of boundary intervals.
$\ShiftAll$ consists of sub-procedures $\AlgName{ShiftPairs}$ and $\AlgName{ShiftElem}$; $\AlgName{ShiftPairs}$ computes the intersections and back-shifting pairwise ($\AlgName{Pairs}(\T)$ enumerates approximations of time intervals in $\T$); $\AlgName{ShiftElem}$ applies the back-shifting.
The results of the procedures may become non-canonical, so $\Normalize$ is applied at last to make them canonical.


$\AlgName{ConsistentAtInitTime}$ is implemented in Algorithm~\ref{a:inittime}.
An input $\T_\phi$ is either $\Universe$, $\emptyset$, or an approximated set;
in the last case, the algorithm picks an earliest approximation with $\AlgName{GetFirstElem}$, and checks whether it contains 0 or not.

\begin{figure*}[t]
\begin{align*}
	\Invert(\T) &:= 
	\begin{cases}
		\makebox[20em][l]{$\emptyset$} & \text{if $\T = \Universe$} \\
		\Universe & \text{if $\T = \emptyset$ ~~} \\
		\Normalize(\ \{(\s,\neg b) ~|~ (\s,b)\in\T \}\ )
		& \text{otherwise~}
	\end{cases} \\
	\JoinOp(\T_1,\T_2) &:= 
	\begin{cases}
		\makebox[20em][l]{$\Universe$} & \text{if $\T_{1} = \Universe \LOr \T_{2} = \Universe$} \\
		\Normalize(\ \T_{1} \cup \T_{2}\ ) & \text{otherwise}
	\end{cases} \\
	\Intersect(\T_1,\T_2) &:=~
	\Invert(\ \JoinOp(\Invert(\T_1), \Invert(\T_2))\ ) \\
	\ShiftAll(\T_1,\T_2) &:=
	\begin{cases}
		\makebox[20em][l]{$\emptyset$} & \text{if $\T_1 = \emptyset \LOr \T_2 = \emptyset$} \\
		\Normalize(\ \AlgName{ShiftPairs}_{\ttt}(\mathbm{T}_1,\mathbm{T}_2)\ )
		& \text{otherwise}
	\end{cases} \\
	\AlgName{ShiftPairs}_{\ttt}(\mathbm{T}_1,\mathbm{T}_2) &:=
	\begin{cases}
		\makebox[21em][l]{$
		\bigl\{ \AlgName{ShiftElem}_\ttt(\P_2) 
		$}
		~~|~~ 
		\P_2 \in \AlgName{Pairs}(\T_2) \bigr\}
		& \text{if $\T_1 = \Universe$} \\
		\makebox[21em][l]{$
		\bigl\{ \AlgName{Intersect}(\ \AlgName{ShiftElem}_\ttt(\P_1),\ \P_1\ ) 
		$}
		~~|~~ \P_1 \in \AlgName{Pairs}(\T_1) \bigr\}
		& \text{if $\T_2 = \Universe$} \\
		\makebox[21em][l]{$
		\bigl\{ \AlgName{Intersect}(\ \AlgName{ShiftElem}_\ttt(\AlgName{Intersect}(\P_1,\P_2)),\ \P_1\ ) 
		$}
		~~|~~ \P_1 \in \AlgName{Pairs}(\T_1), \P_2 \in \AlgName{Pairs}(\T_2) \bigr\} 
		& \text{otherwise}
	\end{cases} \\
	\AlgName{ShiftElem}_\ttt(\T) &:=~
		\Normalize(\ \bigl\{ (\s-\LB{t},\True) ~|~ \Exists{(\s,\True)}{\T} \bigr\} \cup 
		\bigl\{(\s-\UB{t},\False) ~|~ {(\s,\False)}\in{\T} \bigr\}\ ) \\
	\Normalize(\T) &:=
	\begin{cases}
		\textbf{error} & \text{if $\Exists{(\s,\False)}{\T} ~ \s \neq [0] \LAnd \s \ni 0$} \\
		\makebox[20em][l]{$\textbf{error}$} & \text{if $\Exists{(\s,b),(\s',\neg b)}{\T} ~ \s \cap \s' \neq \emptyset$} \\
		\AlgName{N}_4(\AlgName{N}_3(\AlgName{N}_2(\AlgName{N}_1(\T)))) & \text{otherwise}
	\end{cases} \\
	\AlgName{N}_1(\T) &:=~ \bigl\{ (\s,\True)\in\T ~|~ 
		\# \{(\s',\True)\in\T ~|~ \UB{s}' < \LB{s}\} - 
		\# \{(\s'',\False)\in\T ~|~ \UB{s}'' < \LB{s}\} < 1 \bigr\}\ \cup \\
	& \qquad \bigl\{ (\s,\False)\in\T ~|~ 
		\# \{(\s',\True)\in\T ~|~ \UB{s}' < \LB{s}\} - 
		\# \{(\s'',\False)\in\T ~|~ \UB{s}'' < \LB{s}\} < 2 \bigr\} \\
	\AlgName{N}_2(\T) &:=
	\begin{cases}
		\makebox[20em][l]{$\Universe$} & \text{if $\Max\ \T = (\s,\True)$ such that $\UB{s} \leq 0$}\\
		\{ (\s,b)\in\T ~|~ \UB{s} > 0 \} & \text{otherwise}
	\end{cases} \\
	\AlgName{N}_3(\T) &:=~ \bigl\{ (\s,b)\in\T ~|~ \ForAll{(\s',b)}{\T} ~ \s\cap\s' = \emptyset \bigr\} ~\cup~ \bigl\{ (\s\cup\s',b) ~|~ \Exists{(\s,b),(\s',b)}{\T}~ \s\cap\s' \neq \emptyset \bigr\} \\
	\AlgName{N}_4(\T) &:=
	\begin{cases}
		\makebox[20em][l]{$\T \cup \{([0],\True)\}$} & \text{if $(\t,\False) = \Min\ \T$} \\
		\T & \text{otherwise}
	\end{cases} \\
	%
\end{align*}
\caption{\label{f:operations} Procedures for approximated sets of consistent time intervals}
\end{figure*}

\begin{algorithm}[thb]
\caption{\label{a:inittime} $\AlgName{ConsistentAtInitTime}$ algorithm}

\begin{algorithmic}[1]
  \REQUIRE $\T_\phi$
  \ENSURE $\Valid$, $\Unsat$, or $\Unknown$

  \IF{$\T_\phi = \Universe$}
	\RETURN{$\Valid$}
  \ELSIF{$\T_\phi = \emptyset$}
	\RETURN{$\Unsat$}
  \ELSE
    \STATE $(\s,\True) := \AlgName{GetFirstElem}(\T_\phi)$
    \IF{$\UB{s} \leq 0$}
	  \RETURN{$\Valid$}
    \ELSIF{$\LB{s} > 0$}
	  \RETURN{$\Unsat$}
    \ELSEC{$0 \in \s$}
	  \RETURN{$\Unknown$}
    \ENDIF
  \ENDIF
\end{algorithmic}
\end{algorithm}

\subsection{Proof of Lemma~\ref{th:consistent}}
\label{s:canonical:proof}

We check that each condition of a canonical approximation (Definition~\ref{d:approxs}) is assured by $\Normalize$, the last sub-process in each procedure:
\begin{itemize}
\item During the propagation process, polarity alternation might be inhibited by $\JoinOp$, $\Intersect$, or $\ShiftAll$, which locates a boundary interval inside another consistent time interval.
These \emph{embedded} bounds are removed by $\AlgName{N}_1$.
An embodiment can be determined by checking the difference between the numbers of lower and upper bounds in the past since the smallest elements in $\T_1$ and $\T_2$ are always the lower-bound enclosures.
\item The upper bound of each time interval $\s$ in $\T$ becomes non-negative because elements with non-positive upper bounds are filtered out by $\AlgName{N}_2$.
\item No two elements of $\T$ overlap because an overlapping pair with opposite polarity results in an error (the second branch of $\Normalize$) and an overlap with the same polarity is joined ($\AlgName{N}_3$);
thus, the elements in $\T$ can be sorted.
\item $\AlgName{N}_4$ assures that the polarity value of the smallest element is $\True$.
\Qed
\end{itemize}

\subsection{Proof of Theorem~\ref{th:soundness}}
\label{s:soundness:proof}

We perform a structural induction based on the STL formulae.

For the base case $\phi = p \in \APSet_{\phi}$, $\T_p$ exists in $\mathcal{T}$.

For the inductive step, consider STL formulae $\phi_1$ and $\phi_2$, and assume as the inductive hypothesis that we have canonical approximated sets $\T_{\phi_1}$ and $\T_{\phi_2}$ of $T_{\phi_1}$ and $T_{\phi_2}$, respectively.
We show that $\Propagate$ computes the approximated set properly for each formula constructed from $\phi_1$ and $\phi_2$.

When $\phi = \neg \phi_1$, the polarity of each bound of $\T_{\phi_1}$ is switched by $\Invert$ to obtain an approximated set for the complementary time intervals, which is sound regarding the operation~\eqref{e:op:neg} in Step~2 of Section~\ref{s:stl:monitoring}.
Then, $\Normalize$ is applied to canonicalize the result;
it will append or remove the smallest bound.
Let $(\s,\False)$ be the smallest element in a result of polarity inversion.
We confirm that $\Normalize$ is sound in a case analysis:
\begin{itemize}
\item if $\s$ is non-empty and $\s \ni 0$, the computation results in an error (the first branch of $\Normalize$);
\item if $\LB{s} > 0$, the element remains and the element $([0],\True)$ is appended by $\NOp_3$;
\item if $\s = [0]$, the element is removed by $\NOp_2$.
\end{itemize}

When $\phi = \phi_1 \lor \phi_2$, $\T_{\phi_1}$ and $\T_{\phi_2}$ are modified by $\AlgName{Join}$, which joins the elements of both approximated sets; a result might be a non-canonical set when two approximated time intervals from $\T_{\phi_1}$ and $\T_{\phi_2}$ overlap.
Then, $\Normalize$ is applied to unify two overlapping approximations so that the result becomes a sound approximated set with respect to the operation~\eqref{e:op:lor}.
When two approximated time intervals $((\s_1,\True),(\s_1',\False))$ and $((\s_2,\True),(\s_2',\False))$ overlap, the boundary interval (e.g., $\s_1$) either (i) overlaps with another boundary interval, (ii) is included in the inner approximation $(\UB{s}_2,\LB{s}_2')$, or (iii) is excluded from the outer approximation $[\LB{s}_2,\UB{s}_2']$. \Todo{Need a figure?}
We confirm the soundness of $\Normalize$ in another case analysis:
\begin{itemize}
\item in case (i), the bound is removed by the second branch of $\Normalize$ and by $\NOp_3$;
\item in case (ii), the bound is removed by $\NOp_1$; 
\item in case (iii), the bound remains since it should be the bound of the joined time interval.
\end{itemize}

When $\phi = \phi_1 \Until_\ttt \phi_2$, $\T_{\phi_1}$ and $\T_{\phi_2}$ are modified by $\AlgName{ShiftAll}_\ttt$, which applies $\Intersect$, $\AlgName{ShiftPairs}_\ttt$ and $\AlgName{ShiftElem}_\ttt$, those implement the operation~\eqref{e:op:until}.
The soundness of $\AlgName{Intersect}$ with respect to the set intersection is evident because this procedure simply implements the set operation $(\T_1\setminus\PosRealSet \cup \T_2\setminus\PosRealSet)\setminus\PosRealSet$.
$\AlgName{ShiftPairs}_\ttt$ exhaustively applies $\AlgName{ShiftElem}_\ttt$ to each pair of boundary enclosures in $\T_1$ and $\T_2$.
$\AlgName{ShiftElem}_\ttt$ translates the lower and upper bounds, according to the operation~\eqref{e:op:until}; this procedure is sound because an interval enclosure is assumed for each bound of the consistent time intervals.
$\Normalize$, then, resolves the overlaps and closes the lowest bound as in the case of $\phi_1 \lor \phi_2$.
\Qed

\subsection{Computational Complexity of the Operations on Approximated Sets}
\label{s:complexity:op}

Let $\#\T$ be the number of elements in $\T$;
if the bounds appear uniformly in a simulation, $\#\T$ is proportional to $\Norm{\phi}$;
in other words, $\#\T$ is bounded by $\Norm{\phi}/\epsilon^*$ where $\epsilon^*$ is the precision of the floating-point numbers.
The complexity of $\Normalize$ is bounded by $O(\#\T^2)$ since the complexities of $\NOp_1$, $\NOp_2$, $\NOp_3$, and $\NOp_4$ are $O(\#\T^2)$, $O(\#\T)$, $O(\#\T)$, and $O(1)$, respectively.
Without the $\Normalize$ process, the complexities of $\Invert$ and $\JoinOp$ are $O(\#\T)$ and $O(1)$, respectively; together with $\Normalize$, their complexities are $O(\#\T^2)$.
The complexity of $\ShiftAll$ is $O(\#\T^4)$ (let $\#\T$ be the larger cardinality for $\T_1$ or $\T_2$) since the complexities of $\AlgName{ShiftElem}$ and $\AlgName{ShiftPairs}$ are $O(\#\T^2)$ and $O(\#\T^2 \cdot \#\T^2)$, respectively.

\newpage

\profile{Daisuke Ishii}{%
recieved B.Eng, M.Eng, and Ph.D. degrees in computer science from Waseda University, Tokyo, Japan in 2001, 2003, and 2010, respectively.
He was a research fellow at INRIA/LINA in France, from 2010 to 2011, and a research fellow of the JSPS at National Institute of Informatics from 2011 to 2013.
He is currently an Assistant Professor at Tokyo Institute of Technology.
His research interests include interval analysis and formal methods for hybrid systems.

Tokyo Institute of Technology, Department of Computer Science,
2-12-1-W8-67, Ookayama, Meguro-ku, Tokyo, 152-8550 Japan.
}
\label{profile}


\profile{Naoki Yonezaki}{%
Naoki Yonezaki currently is Visiting Professor of Open University of Japan and Emeritus Professor of Tokyo Institute of Technology.
He has been Professor of Tokyo Institute of Technology since 1991. He was also  Professor of Japan Advanced Institute of Science and Technology from1991 till 1995. His research interests include verification of software specification, verification of security and formal approach to system biology.
He is a member of IEICE, IPSJ, JSSST, JSAI, ACM and EATCS. He awarded to a fellowship from JSSST in 2008.

Open University of Japan,
2-11 Wakaba, Mihama-ku, Chiba 261- 8586 Japan.
}

\profile{Alexandre Goldsztejn}{%
After receiving his Ph.D. from the University of Nice in 2005, he has spent one year as a postdoctoral fellow in USA, in the University of Central Arkansas and in the University of California Irvine. He was then granted a tenured research associate position at CNRS, where he continued his researches on numerical constraint programming, with emphasis on quantified constraints and positive dimensional manifolds, global optimization and dynamical systems.

IRCCyN -- Ecole Centrale de Nantes,
1, rue de la No\"{e}, BP 92101, 44321 Nantes Cedex~3, France.
}

\end{document}